\documentclass[11pt]{article}

\usepackage[margin = 1in]{geometry}
\usepackage{graphicx}              
\usepackage{amsmath}               
\usepackage{amsfonts}              
\usepackage{amsthm}                
\usepackage{amssymb}
\usepackage{mathrsfs}
\usepackage{url}
\usepackage{xcolor}
\usepackage{hyperref}
\usepackage{algpseudocode}
\usepackage[plain]{algorithm}
\usepackage{enumitem}
\usepackage{tcolorbox}
\usepackage{mathtools}
\usepackage{tabularx}
\usepackage{ltablex}
\usepackage{bm}
\usepackage{dsfont}
\usepackage{titling}
\usepackage{wasysym}
\usepackage{marvosym}
\usepackage{tikz}
\usetikzlibrary{quantikz}

\hypersetup{
	unicode = true,
	colorlinks = true,
	citecolor = blue,
	filecolor = blue,
	linkcolor = blue,
	urlcolor = blue,
	pdfstartview = {FitH},
}

\theoremstyle{plain}
\newtheorem{theorem}{Theorem}[section]
\newtheorem{lemma}[theorem]{Lemma}

\newtheorem{proposition}[theorem]{Proposition}

\theoremstyle{definition}

\newtheorem{problem}[theorem]{Problem}
\newtheorem{fact}[theorem]{Fact}

\algrenewcommand{\Require}{\item[\textit{Input:}]}
\algrenewcommand{\Ensure}{\item[\textit{Output:}]}
\algrenewcommand{\alglinenumber}[1]{#1.}

\newcommand{\quat}{\mathcal{B}_{p, \infty}}
\newcommand{\sheaf}[1]{\mathscr{#1}}

\newcommand{\bolt}[1]{\text{\Lightning}_{#1}}

\newcommand{\Romnum}[1]{\uppercase\expandafter{\romannumeral #1}}

\DeclareMathOperator{\trd}{trd}
\DeclareMathOperator{\tr}{Tr}
\DeclareMathOperator{\cl}{Cl} 
\DeclareMathOperator{\nrd}{nrd} 

\DeclareMathOperator{\poly}{poly}

\DeclareMathOperator{\spec}{Spec}
\let\hom\relax
\DeclareMathOperator{\hom}{Hom}

\DeclarePairedDelimiter{\abs}{\lvert}{\rvert}
\let\ket\relax
\DeclarePairedDelimiter{\ket}{\lvert}{\rangle}
\let\bra\relax
\DeclarePairedDelimiter{\bra}{\langle}{\rvert}
\let\braket\relax
\DeclarePairedDelimiterX{\braket}[2]{\langle}{\rangle}{#1 \delimsize\vert #2}
\DeclarePairedDelimiter{\lrang}{\langle}{\rangle}
\DeclarePairedDelimiter{\opnorm}{\lVert}{\rVert}

\def\A{\mathcal{A}}
\def\Q{\mathbb{Q}}
\def\C{\mathbb{C}}

\def\R{\mathbb{R}}
\def\Z{\mathbb{Z}}
\def\F{\mathbb{F}}
\def\P{\mathbb{P}}
\def\H{\mathbb{H}}
\def\T{\mathbb{T}}

\def\bank{\mathsf{Bank}}
\def\ver{\mathsf{Ver}}
\def\gen{\mathsf{Gen}}

\def\X{\mathcal{X}}
\def\O{\mathcal{O}}

\title{Cryptanalysis of Three Quantum Money Schemes}

\thanksmarkseries{arabic}
\author{
Andriyan Bilyk\thanks{Ryerson University, {\tt andriyan.bilyk@ryerson.ca}}
\and
Javad Doliskani\thanks{Ryerson University, {\tt javad.doliskani@ryerson.ca}}
\and
Zhiyong Gong\thanks{Ryerson University, {\tt zhiyong.gong@ryerson.ca}}
}

\date{}
\allowdisplaybreaks

\begin{document}
\maketitle

\begin{abstract}
    We investigate the security assumptions behind three public-key quantum money schemes. Aaronson and Christiano proposed a scheme based on hidden subspaces of the vector space $\mathbb{F}_2^n$ in 2012. It was conjectured by Pena et al in 2015 that the hard problem underlying the scheme can be solved in quasi-polynomial time. We confirm this conjecture by giving a polynomial time quantum algorithm for the underlying problem. Our algorithm is based on computing the Zariski tangent space of a random point in the hidden subspace.
    
    Zhandry proposed a scheme based on multivariate hash functions in 2017. We give a polynomial time quantum algorithm for cloning a money state with high probability. Our algorithm uses the verification circuit of the scheme to produce a banknote from a given serial number.
    
    Kane, Sharif and Silverberg proposed a scheme based on quaternion algebras in 2021. The underlying hard problem in their scheme is cloning a quantum state that represents an eigenvector of a set of Hecke operators. We give a polynomial time quantum reduction from this hard problem to a linear algebra problem. The latter problem is much easier to understand, and we hope that our reduction opens new avenues to future cryptanalyses of this scheme.
\end{abstract}

\newpage

\section{Introduction}

The first quantum money scheme was proposed by Wiesner \cite{wiesner1983conjugate} around 1970. The idea was to represent banknotes using quantum states which, by the laws of quantum mechanics, would be impossible to copy. Wiesner's scheme, which is nowadays called a private-key quantum money scheme, worked as follows. The bank generates a pair $(s, \rho_s)$ consisting of a classical serial number $s$ and a quantum state $\rho_s$, and stores $s$ along with a classical description of $\rho_s$ in a database. Without access to the database, no algorithm would be able to copy $\rho_s$ with non-negligible probability. Therefore, only the issuing bank would be able to verify the banknotes. Although this scheme is information-theoretically secure, it suffers major drawbacks in terms of practicality: First, the issuing bank has to maintain a huge database containing classical descriptions of every banknote ever produced. Second, only the bank can verify a banknote, and therefore the users have to take the money back to the bank for every transaction.

Subsequent efforts to address the drawbacks of Wiesner's scheme led to the idea of \textit{public-key quantum money}. The first formal treatment of public-key quantum money was given by Aaronson \cite{aaronson2009quantum}. Intuitively, in such a scheme, a banknote can be verified by everyone, using a publicly known algorithm, but cannot be copied by anyone. The first explicit construction of public-key quantum money, which was based on stabilizer states, was given in \cite{aaronson2009quantum}. The construction was later proved to be insecure by Lutomirski et al. \cite{lutomirski2009breaking}. 

Public-key quantum money cannot be information-theoretically secure, it has to be based on a computational assumption. There have been many attempts at building a secure public-key quantum money scheme. Farhi et al. \cite{farhi2012quantum} proposed a scheme based on knot theory. A banknote in their scheme consists of a quantum state $\rho_s$ representing a superposition of certain grid diagrams, and serial number $s$ which is the Alexander polynomial of the grids in $\rho_s$. A banknote is verified using a procedure based on a classical Markov chain. 

Aaronson and Christiano \cite{aaronson2012quantum} proposed a scheme based on hidden subspaces of the vector space $\F_2^n$. In their scheme, a banknote consists of a quantum state $\rho_s$ that is a uniform superposition of elements in a random linear subspace $A_s \subset \F_2^n$ of dimension $n / 2$, and a serial number $s$ that provides membership oracles for $A_s$ and the orthogonal complement $A_s^\perp$. The verification of a banknote is done using a projection operator built from the membership oracles for $A_s$ and $A_s^\perp$.

Zhandry \cite{zhandry2019quantum} adapted the idea proposed in \cite{lutomirski2009breaking} to give a construction of public-key quantum money based on multivariate hash functions. In their scheme, even the issuing bank cannot reproduce the same banknote. A bank note in their scheme consists of a quantum state that is a uniform superposition of the preimages of a random output $y$ of a hash function, and a serial number which is the point $y$. To be able to verify banknotes, a restricted class of hash functions are used. The verification is done using the quantum Fourier transform and computing the ranks of matrices in a certain superposition.

Kane, Sharif and Silverberg \cite{kane2021quantum} proposed a completely different construction for public-key quantum money based on quaternion algebras. There exists a set of commuting operators, called Hecke operators, on the space of modular forms. The Hecke operators are Hermitian, and therefore the space of modular forms has a basis consisting of simultaneous eigenvectors for these operators. A banknote in this scheme consists of a quantum state $\rho_s$ that is an eigenstate of the Hecke operators, and a serial number $s$ that is a set of eigenvalues corresponding to $\rho_s$ and a certain set of Hecke operators. The verification of a banknote is done using phase estimation.

\subsection{Previous cryptanalysis}

Unfortunately, the security of all proposed public-key quantum money schemes to this date have been based on ad hoc/nonstandard computational assumptions. However, there has been a limited number of attempts, many of which were not convincing, at cryptanalyzing these assumption. As mentioned above, the public-key quantum money scheme of \cite{aaronson2009quantum} was broken by Lutomirski et al. \cite{lutomirski2009breaking} who presented two attacks that worked in different parameter regimes. Their first attack was classical and used only the description of the verification circuit, while their second attack was a quantum algorithm that could generate banknotes that were different from the intended banknotes but could pass verification with non-negligible probability. 

Pena et al. \cite{pena2015algebraic} presented an attack on the scheme \cite{aaronson2012quantum}, but their attack was not on the original parameters. In particular, the explicit construction proposed in \cite{aaronson2012quantum} is defined over $\F_2$, but the attack presented in \cite{pena2015algebraic} is only efficient if the scheme was defined over $\F_p$ for a larger prime $p$. Pena et al. conjectured that for $p = 2$ their algorithm would run in quasi-polynomial time. Their attack is entirely classical and uses only the serial number given in the banknote.

Roberts \cite{roberts2021security} proposed an attack on the assumption underlying the scheme of \cite{zhandry2019quantum} for a certain parameter regime. However, as discussed in a later version of \cite{zhandry2019quantum}, the attack does not make the system insecure. Moreover, the hardness assumption behind the scheme could be modified in a way that renders the attack useless.

\subsection{This work}

In this paper, we carry out cryptanalysis on the following three schemes:

\begin{itemize}[leftmargin = *]
    \item Aaronson and Christiano \cite{aaronson2012quantum}. As mentioned above, this scheme is based on hidden subspaces of the vector space $\F_2^n$. The hard problem behind this scheme is to recover a secret linear subspace of $\F_2^n$ from the the set of common roots of a set of random polynomials over $\F_2$. The previous attack, carried out by Pena et al. \cite{pena2015algebraic}, against this scheme uses the Gr\"{o}bner basis algorithm. Their attack is only effective if the scheme was defined over $\F_p$ for larger $p$, instead of $\F_2$. It was conjectured in \cite{pena2015algebraic} that there exists a quasi-polynomial time algorithm for the original scheme over $\F_2$. 
    
    We give a polynomial time quantum algorithm for the case $\F_2$. Our attack is geometric, it uses the fact that a linear variety is isomorphic to its Zariski tangent space at any point. In particular, the subspace hidden by a set of polynomials can be viewed as a linear variety $Y$ included in an affine variety $X$. We show that one can recover the hidden subspace by simply computing the kernel of the Jacobian matrix of $X$ at a random point of the subspace $Y$. Such a random point is obtained by measuring a money state.
    
    \item Zhandry \cite{zhandry2019quantum}. This scheme is based on Multivariate Hash Functions. There are different versions of the paper \cite{zhandry2019quantum} that give slightly different formulations of the scheme, but the underlying hardness assumptions can all be stated using the idea of multi-collision of multivariate polynomial hash functions. We will focus on the version \cite{zhandry2017qlightning}, in which the hardness assumption is directly stated using quadratic forms.
    
    We give a polynomial time quantum algorithm, using which one can obtain arbitrarily many copies of a given banknote. Our algorithm uses the verification circuit of the scheme to produce a certain superposition of ``bolts''. A measurement on this superposition then produces a copy of the money state with high probability. The ability to produce arbitrarily many copies of the money state comes from the fact the measurement outcome $y'$ can be compared against the given serial number $y$. If $y' = y$ then desired money state has been produced, otherwise the whole process is repeated.  
    
    \item Kane, Sharif and Silverberg \cite{kane2021quantum}. This scheme is based on quaternion algebras. There have not been any attempts at cryptanalyzing this scheme yet, partly because of the mathematics involved which is rather nontrivial. The underlying hard problem in this scheme is simply cloning the money state which is an eigenstate of the set of all Hecke operators. 
    
    Although we have not been able to break this scheme, we give a polynomial time quantum reduction from the underlying hard problem to the problem of inverting a certain matrix $A$. Our algorithm is based on a deep connection between ideal classes in a rational quaternion algebra and the space of modular forms of weight $2$ over the complex numbers. The connection is provided by the theory of Theta series. The matrix $A$ is highly structured, with entries the representation numbers of explicit quaternary quadratic forms. In other words, $A$ can be described by a set of polynomials of degree two in four variables. To work with this matrix, one does not need to know much about the theory behind the scheme. We believe that this reduction is an important first step toward an effective cryptanalysis of this scheme, and opens interesting new avenues for future attacks.
\end{itemize}


\section{Preliminaries}

\paragraph{Quantum computation.}
For a detailed treatment on quantum information we refer the reader to \cite{watrous2018theory}. A register in this paper can contain a classical or a quantum state. The classical state of a register $\mathsf{X}$ is described by a finite alphabet $\Sigma$. An example alphabet used in this paper is $\Sigma = \F_2^n$. The quantum state of $\mathsf{X}$ is described by the set of density operators $\mathrm{D}(\X)$ where $\X$ is the complex Euclidean space $\C^\Sigma$.

We will also represent quantum states using unit vectors in $\X$, which are called pure states. We will use the Dirac notation when describing unit vectors. In particular, a unit column vector $x \in \X$ is written as $\ket{x}$, and the row vector $x^*$, which is the conjugate transpose of $x$, is written as $\bra{x}$. The quantum state of the register $\mathsf{X}$ can then be written as
\[ \sum_{x \in \Sigma} \alpha_x\ket{x}, \quad \sum_{x \in \Sigma} \abs{\alpha_x}^2 = 1. \]
The trace distance between two quantum states $\rho$ and $\sigma$ is defined by $\opnorm{\rho - \sigma}_1$ where $\opnorm{A}_1 = \tr \sqrt{A^*A}$ is the trace norm of linear operators. For pure states $\ket{\psi}$ and $\ket{\phi}$ the distance is computed using
\[ \opnorm{\ket{\psi}\bra{\psi} - \ket{\phi}\bra{\phi}}_1 = 2 \sqrt{1 - \abs{\braket{\psi}{\phi}}^2}. \]
For quantum state $\rho_i$ and $\sigma_i$, $i = 1, 2$
\[ \opnorm{\rho_1 \otimes \rho_2 - \sigma_1 \otimes \sigma_2}_1 \le \opnorm{\rho_1 - \sigma_1}_1 + \opnorm{\rho_2 - \sigma_2}_1. \]

\paragraph{Public-key quantum money.}
Following the definition in \cite{aaronson2012quantum}, a public-key quantum money scheme consists of three algorithms:
\begin{itemize}[font = \normalfont]
    \item $\gen$. Takes as input a security parameter $\kappa$, and generates a key pair $(k_{\text{pri}}, k_{\text{pub}})$ in probabilistic polynomial time (in $\kappa$). Here, $k_{\text{pri}}$ is the private key and $k_{\text{pub}}$ is the public key.
    \item $\bank$. Takes as input a private key $k_{\text{pri}}$, and generates a banknote in probabilistic polynomial time.
    \item $\ver$. Takes as input the public key $k_{\text{pub}}$ and a banknote, and outputs either `accept' or `reject'.
\end{itemize}
A banknote, denoted by $\ket{\$}$, usually consists of a pair $(s, \rho_s)$ where $s$ is a binary string and $\rho_s$ is a mixed quantum state. The string $s$ is called the \textit{serial number}, and the state $\rho_s$ is called the \textit{money state}. 

The notion of security for a quantum money scheme is best explained using an adversary-challenger game. The goal for an adversary $A$ is to produce a copy of a valid banknote without having access to the private key. Consider the following game played by $A$ and a challenger. The challenger calls $\bank$ to generate a banknote $\ket{\$} = (s, \rho_s)$, and sends the banknote to $A$. $A$ generates two alleged banknotes $\ket{\$_1} = (s_1, \rho_{s_1})$ and $\ket{\$_2} = (s_2, \rho_{s_2})$ where $\rho_{s_1}$ and $\rho_{s_2}$ are allowed to be entangled. The challenger accepts if and only if $\ver$ accepts both $\ket{\$_1}$ and $\ket{\$_2}$ and $s_1 = s_2 = s$. A quantum money scheme $(\gen, \bank, \ver)$ is said to be secure if for any polynomial time quantum adversary $A$, the challenger accepts in the above experiment with probability negligible in $\kappa$.


\section{Quantum Money From Hidden Subspaces}

In this section, we investigate the security of the quantum money scheme proposed by Aaronson and Christiano \cite{aaronson2012quantum}, which is based on hidden linear subspaces of $\F_2^n$.

\subsection{The quantum money scheme}

Let $n = \kappa$ where $\kappa$ is the security parameter. The $n$-qubit money state in this scheme is a uniform superposition 
\[ \ket{A} = \frac{1}{2^{n / 4}} \sum_{x \in A} \ket{x}, \]
where $A \subset \F_2^n$ is a random linear subspace of dimension $n / 2$. The bank generates such a state by choosing a set of $n / 2$ random and linearly independent elements of $\F_2^n$ and building a superposition from their linear combinations. The state $\ket{A}$ is then verified by querying to classical oracles $U_A$ and $U_{A^\perp}$ where $A^\perp$ is the orthogonal complement of $A$. These oracles are membership oracles. More precisely, for $x \in \F_2^n$
\[ U_A\ket{x} =
\begin{cases}
    -\ket{x} & \text{if } x \in A, \\
    \ket{x} & \text{otherwise}.
\end{cases}
\]
The oracle $U_{A^\perp}$ is defined similarly. Using these oracles we can build projectors $\P_A$ and $\P_{A^\perp}$ onto the basis elements of $A$ and $A^\perp$, respectively. The verification process works by first projecting onto $A$, then applying the $n$-qubit Hadamard transform $H^{\otimes n}$, then projecting on $A^\perp$ and finally applying another $H^{\otimes n}$. So the verification algorithm can be written as $V_A := H^{\otimes n} \P_{A^\perp} H^{\otimes n} \P_A$. It can be shown that $V_A = \ket{A}\bra{A}$, i.e., $V_A$ is just a projector onto $A$. This means that the verification algorithm accepts a state $\ket{\psi}$ with probability $\abs{\braket{\psi}{A}}^2$. 

In the above scheme, the following algorithms are publicly accessible.
\begin{description}[font = \normalfont]
    \item [$\mathsf{G}(r)$:] takes a random $r \in \{ 0, 1 \}^n$ as input. Outputs a set of random linearly independent elements $\{ y_1, \dots, y_{n / 2} \}$ for a subspace $A_r \subset \F_2^n$, and a serial number $s_r$.
    \item [$\mathsf{H}(s)$:] checks if $s$ is a valid serial number.
    \item [$\mathsf{T}(r)$:] takes an input of the form $\ket{s}\ket{x}$ and applies $U_{A_r}$ to $\ket{x}$ if $s$ is a valid serial number.
    \item [$\mathsf{T}^\perp(r)$:] similar to $\mathsf{T}(r)$ but with $U_{A_r}$ replaced by $U_{A_r^\perp}$.
\end{description}
The $\bank$ and $\ver$ algorithms are then defined as follows:
\begin{description}[font = \normalfont]
    \item [$\bank$] calls $\mathsf{G}(r)$ for a uniformly random $r \in \{ 0, 1 \}^n$ to receive a pair $(s_r, A_r)$, and outputs a banknote $\ket{\$_r} = \ket{s_r}\ket{A_r}$.
    \item [$\ver$] takes as input a banknote $\cent$, and uses $\mathsf{H}$ to check if $\cent$ is of the form $(s, \rho)$ for a valid serial number $s$. If $s$ is valid then it uses $\mathsf{T}$ and $\mathsf{T}^\perp$ to compute $V_A(\rho)$, and outputs `accept' if and only if $V_A$ outputs `accepts'.
\end{description}
Assuming black box access to the oracles $U_A$ and $U_{A^\perp}$, it was proved in \cite{aaronson2012quantum} that a counterfeiter needs at least $\Omega(\sqrt{\epsilon} 2^{n / 4})$ queries to prepare a state $\rho$ such that $\bra{A}^{\otimes 2} \rho \ket{A}^{\otimes 2} \ge \epsilon$. However, there has been no provably secure instantiation of these oracles to this date. 

The first instantiation of $U_A$ and $U_{A^\perp}$ was proposed in the same paper \cite{aaronson2012quantum}, which was based on multivariate polynomials. The idea is to ``obfuscate'' the subspace $A$ using a set of polynomials $p_1, \dots, p_m \in \F_2[x_1, \dots, x_n]$. More precisely, let $I_{d, A}$ be the set of polynomials of degree $d$ in $\F_2[x_1, \dots, x_n]$ that vanish on $A$. Then the polynomials $p_i$ are selected uniformly at random from $I_{d, A}$. For a large enough $m$, the polynomials $p_1, \dots, p_m$ uniquely determine $A$ with overwhelming probability. This means for all $v \in A$ we will have $p_1(v) = \cdots = p_m(v) = 0$, and for any $v \notin A$ we will have $p_k(v) \ne 0$ for at least one $1 \le k \le m$ with overwhelming probability. This is proved in the following lemma.
\begin{lemma}[{\cite[Lemma 29]{aaronson2012quantum}}]
    \label{lem:poly-orac}
    Let $A \subseteq \F_2^n$ be a linear subspace and let $\beta > 1$ be a constant. For any set of polynomials $p_1, \dots, p_{\beta n} \in I_{d, A}$ selected uniformly at random, define $Z = \{ v \in \F_2^n : p_i(v) = 0 \text{ for all } 1 \le i \le \beta n \}$. Then $A \subseteq Z$, and $\Pr[A = Z] = 1 - 2^{-\Omega(n)}$.
\end{lemma}
One can sample uniformly at random from $I_{d, A}$ in time $O(n^d)$ which is polynomial when $d$ is a constant. An immediate consequence of Lemma \ref{lem:poly-orac} is that the set of random polynomials $p_1, \dots, p_{\beta n}$ can be effectively used as a membership oracle for $A$. Therefore, the banknote $(s_r, \ket{A_r})$, which can be generated efficiently by the bank, consists of a superposition over the random subspace $A_r$ and the list of polynomials $s_r = \{ p_i \}_{1 \le i \le \beta n}$ corresponding to $A_r$. The hardness assumption underlying the quantum money scheme of \cite{aaronson2012quantum} is based on the following problem:
\begin{problem}
    \label{prb:hidden-subsp}
    Let $d \ge 3$ be a constant integer and $\beta > 1$ a real number. Let $A \subset \F_2^n$ be a random linear subspace of dimension $n / 2$ and let $p_1, \dots, p_{\beta n} \in I_{d, A}$ be selected uniformly at random. Given the banknote $(s, \ket{A})$, where $s = \{ p_i \}_{1 \le i \le \beta n}$, recover a set of generators for $A$.
\end{problem}
A noisy version of Problem \ref{prb:hidden-subsp} was also introduced in \cite{aaronson2012quantum}, but was later shown to be quantum polynomial time equivalent to the noiseless version \cite[Section 9.6]{aaronson2016complexity}. In this paper, we give a quantum polynomial time algorithm for Problem \ref{prb:hidden-subsp}. Our idea is based on computing the Zariski tangent space of the variety described by the set of polynomials $s$.

\subsection{Background}

In this section, we review some basic facts about tangent spaces in algebraic geometry. Let $X$ be a scheme, and let $x \in X$ be any point. Let $\kappa(x) = \sheaf{O}_{X, x} / \mathfrak{m}_x$ be the residue field at $x$, where $\sheaf{O}_{X, x}$ is the local ring at $x$ and $\mathfrak{m}_x$ is the maximal ideal of $\sheaf{O}_{X, x}$. Then $\mathfrak{m}_x / \mathfrak{m}_x^2$ is a vector space over $\kappa(x)$. The Zariski tangent space of $X$ in $x$ is defined as the dual vector space 
\[ T_xX = (\mathfrak{m}_x / \mathfrak{m}_x^2)^\vee. \]
The following are standard in algebraic geometry, but include some explanations for the sake of completeness.
\begin{fact}
    \label{fct:tang-incl}
    If $Y \subseteq X$ is a closed subscheme, and $y \in Y$, then $T_yY \subseteq T_yX$. This is because the morphism
    \[ \mathfrak{m}_y / \mathfrak{m}_y^2 \rightarrow \mathfrak{m}_x / \mathfrak{m}_x^2 \]
    of $\kappa(y)$-vector spaces, where $\mathfrak{m}_y$ is the maximal ideal of $\sheaf{O}_{Y, y}$, is surjective. Therefore, the dual morphism is injective.
\end{fact}
\begin{fact}
    \label{fct:tang-int}
    For closed subschemes $Y, Z \subseteq X$ and a point $y \in Y \cap Z$, we have $T_y (Y \cap Z) = T_yY \cap T_yZ$. Here, $Y \cap Z$ is the scheme-theoretic intersection. Locally, we can replace $X$ with a $\spec(A)$, for a ring $A$, so that $Y$ and $Z$ are described by some ideals $I, J \subset A$. Then the intersection $Y \cap Z$ is $\spec A / (I + J)$, and the above statement can be proved using straightforward commutative algebra.
\end{fact}
Let $k$ be a field. By a $k$-variety $X$ we mean a reduced separated scheme of finite type over $k$. In this paper, we are interested in affine varieties, i.e., in schemes of the form $X = \spec k[T_1, \dots, T_n] / I$ where $I = (f_1, \dots, f_m)$ is an ideal in the polynomial ring $k[T_1, \dots, T_n]$. A $k$-variety is called linear if it is defined by the intersection of hyperplanes. Therefore, $X$ is linear when the defining polynomials $f_1, \dots, f_m$ all have degree one. A point $x \in X$ is called a $k$-valued point if $\kappa(x) = k$. When $X$ is affine, the Zariski tangent space at a $k$-valued point $x$ can be computed using the Jacobian matrix at $x$. The Jacobian matrix of the set of polynomials $f_1, \dots, f_m$ at $x$ is defined as
\[ J_{f_1, \dots, f_m}(x) = 
\begin{bmatrix}
    \dfrac{\partial f_1}{\partial T_1}(x) & \cdots & \dfrac{\partial f_1}{\partial T_n}(x) \\
    \vdots & \ddots & \vdots \\
    \dfrac{\partial f_m}{\partial T_1}(x) & \cdots & \dfrac{\partial f_m}{\partial T_n}(x)
\end{bmatrix},
\]
where the derivatives $\partial f_i / \partial T_j$ are formal derivatives. 

For the affine variety $X$ defined by the ideal $I$ as above, the tangent space $T_xX$ at the point $x = (x_1, \dots, x_n) \in X$ is the linear variety defined by the coordinates of the vector
\begin{equation}
    \label{equ:tang-spc}
    J_{f_1, \dots, f_m}(x) 
    \begin{bmatrix}
        T_1 - x_1 \\
        \vdots \\
        T_n - x_n
    \end{bmatrix}.
\end{equation}
In other words, letting 
\[ g_i = \sum_{j = 1}^n \frac{\partial f_i}{\partial T_j}(x)(T_j - x_j),\]
the space $T_xX$ is defined as the zero set of the ideal $(g_1, \dots, g_m)$. It follows from the above that 
\begin{fact}
    \label{fct:tang-lin}
    If $X$ is a linear $k$-variety and $x \in X$ is a $k$-valued point, then $T_xX = X$.
\end{fact}

\subsection{Cryptanalysis}
\label{sec:cryptnls-hdn-spc}

In this section, we give a polynomial time quantum algorithm for Problem \ref{prb:hidden-subsp}. We are given a banknote $(s, \ket{A})$ where $A \subset \F_2^n$ is a random linear subspace of dimension $n / 2$, and $s$ is a set of polynomials $p_1, \dots, p_m \in I_{d, A}$, where $m = \beta n$, $\beta > 1$. The first step is to measure the state $\ket{A}$ to obtain uniformly random element $x \in A$. So now we have two pieces of classical information: the set of polynomials $\{ p_i \}$ and a uniformly random $x \in A$.

Geometrically, the subspace $A$ is the set of $\F_2$-valued points in a linear variety $Y$ given by an ideal $(H_1, \dots, H_{n / 2})$ with $H_i$'s homogeneous linear polynomials. It follows from \eqref{equ:tang-spc} and Fact \ref{fct:tang-lin} that $T_xY = Y$. In fact, since the $H_i$ are homogeneous, we have
\[ Y = \ker J_{H_1, \dots, H_{n / 2}}(x). \]
Define the ideal $I = (p_1, \dots, p_m)$, and let $X = \spec \F_2[T_1, \dots, T_n] / I$. By Lemma \ref{lem:poly-orac}, we can safely assume that the varieties $Y$ and $X$ have the same set of $\F_2$-valued points. If we define $X_i = \spec \F_2[T_1, \dots, T_n] / (p_i)$ then $X = X_1 \cap \cdots \cap X_m$. 

Since $p_i \in I_{d, A}$, we have $Y \subseteq X_i$ for all $i = 1, \dots, m$, and therefore, by Fact \ref{fct:tang-incl}, 
\[ Y = T_xY \subseteq T_xX_i, \quad i = 1, \dots, m. \]
From this and Fact \ref{fct:tang-int} we obtain
\begin{equation}
    \label{equ:tang-intrs}
    Y \subseteq T_xX_1 \cap \cdots \cap T_xX_m = T_xX.
\end{equation}
We prove that equality holds in \eqref{equ:tang-intrs} with high probability.
\begin{proposition}
    \label{prp:tang-intrs}
    We have $\Pr[Y = T_xX] = 1 - 2^{-\Omega(n)}$ in \eqref{equ:tang-intrs}.
\end{proposition}
\begin{proof}
    Note that the tangent space $T_xX_i$ is given by the points on the hyperplane $h_i = \sum_{j = 1}^n (\partial p_i / \partial T_j)(x)(T_j - x_j)$. Since $Y \subseteq T_xX_i$, we have $h_i \in I_{1, A}$. We first prove the following claim.
    
    \textit{Claim}. For any $v \notin A$, exactly half of the elements in $I_{1, A}$ vanish at $v$.
    
    The claim can be proved using the same trick as in the proof of Lemma \ref{lem:poly-orac}. More precisely, there exists a $w = (w_1, \dots, w_n) \in A^\perp$ such that $w \cdot v = 1$, so for every $h \in I_{1, A}$ we also have $\tilde{h} = h + w_1T_1 + \cdots + w_nT_n \in I_{1, A}$. The claim now follows from the fact that exactly one of the $h$ and $\tilde{h}$ vanish at $v$. 
    
    Suppose, for now, that the hyperplanes $T_xX_i$, or equivalently the $w_i$, are uniformly random among the hyperplanes containing $Y$. Let $v \notin A$. Then it follows from the above claim that $v \in T_xX_1 \cap \cdots \cap T_xX_m$ with probability $2^{-m}$, which prove the proposition. So it remains to prove that the hyperplanes $w_i \in I_{1, A}$ are uniformly random when the $p_i \in I_{d, A}$ are uniformly random. Define the mapping    
    \[
    \begin{array}{rrll}
        D_x : & I_{d, A} & \longrightarrow & I_{1, A} \\
        & f & \longmapsto & \displaystyle \sum_{j = 1}^n \frac{\partial f}{\partial T_j}(x)(T_j - x_j).
    \end{array}
    \]
    Since translation by a point of $A$ preserves $I_{d, A}$ and $I_{1, A}$, we may assume that $x$ is the origin. In that case, the mapping $D_x$ sends a polynomial $f \in I_{d, A}$ to its linear part. That means if we write $f = g + h$ where $\deg(h) \le 1$ then $D_x(f) = h \in I_{1, A}$. Now if $f$ is uniformly random then every monomial of $f$ is uniformly random (among the monomials vanishing on $A$) and therefore, the linear part of $f$ is also uniformly random.
\end{proof}

To summarize, the hidden subspace can be recovered using the following algorithm.

\begin{algorithmic}[1]
    \Require A banknote $(s, \ket{A})$
    \Ensure A set of generators for the subspace $A$
    \State Measure the state $\ket{A}$ to obtain a uniformly random $x \in A$.
    \State Compute the Jacobian matrix $J_s(x)$ of the set of polynomials $s$ at $x$.
    \State Return a set of generators for the kernel of $J_s(x)$. 
\end{algorithmic}

\subsection{Example}

To explain the idea of the algorithm of Section \ref{sec:cryptnls-hdn-spc}, we give a concrete step-by-step example of the computation. The example is generated randomly over the space $\F_2^8$, so $n = 8$. Let $A \subset \F_2^8$ be a random linear subspace of dimension $n / 2 = 4$ generated by the rows of the following random full rank $4 \times 8$ matrix:
\[
A = \mathrm{Img}
\begin{bmatrix}
    1 & 0 & 0 & 0 & 1 & 1 & 0 & 0 \\
    0 & 1 & 0 & 0 & 1 & 0 & 1 & 0 \\
    0 & 0 & 1 & 0 & 1 & 0 & 1 & 1 \\
    0 & 0 & 0 & 1 & 1 & 0 & 0 & 1 
\end{bmatrix}.
\]
Then $A$ is the set of $\F_2$-valued points of the linear variety given by the ideal $(H_1, H_2, H_3, H_4)$ where
\begin{align*}
    H_1(T) & = T_1 + T_6 \\
    H_2 (T) & = T_2 + T_5 + T_6 + T_8 \\
    H_3 (T) & = T_3 + T_5 + T_6 + T_7 + T_8 \\
    H_4(T) & = T_4 + T_5 + T_6 + T_7
\end{align*}
If we set $\beta = 9 / 8$ then $m = \beta n = 9$. Let $p_1, p_2, \dots, p_9 \in I_{3, A}$ be the following polynomials chosen uniformly at random.

\vspace*{2mm}

{
\setlength{\tabcolsep}{1mm}
\renewcommand{\arraystretch}{1.3}
\begin{tabularx}{\textwidth}{lX}
    $p_1(T) = $ & $T_1T_3^2 + T_2T_3^2 + T_1^2T_4 + T_1T_3T_5 + T_2T_4T_5 + T_2T_3T_6 + T_3^2T_6 + T_1T_4T_6 + T_3T_5T_6 + 
    T_2T_3T_7  + T_2T_3T_8 + T_1T_5T_8 + T_2T_5T_8 + T_4T_5T_8 + T_5^2T_8 + T_4T_7T_8 + T_6T_7T_8 + 
    T_7^2T_8 + T_1T_2 + T_1T_5 + T_1T_6 + T_2T_6 + T_5T_6 + T_6^2 + T_1T_8 + T_6T_8 + T_2 + T_5 + T_6 + T_8$ \\
    
    $p_2(T) =$ & $T_1T_3T_4 + T_1T_3T_5 + T_3T_4T_5 + T_4T_5^2 + T_1T_3T_6 + T_3^2T_6 + T_3T_5T_6 + T_4T_5T_6 + T_1T_6^2 +
    T_3T_6^2 + T_6^3 + T_1T_3T_7 + T_4T_5T_7 + T_3T_6T_7 + T_2T_5T_8 + T_4T_5T_8 + T_5^2T_8 + T_3T_6T_8 + 
    T_5T_6T_8 + T_3T_7T_8 + T_5T_7T_8 + T_6T_7T_8 + T_7^2T_8 + T_5T_8^2 + T_7T_8^2 + T_1T_2 + T_2T_6 + 
    T_4T_6 + T_5T_6 + T_6^2 + T_6T_7 + T_1T_8 + T_6T_8 + T_3 + T_5 + T_6 + T_7 + T_8$ \\
    
    $p_3(T) =$ & $T_2T_3T_5 + T_3T_5^2 + T_2^2T_6 + T_1T_5T_6 + T_2T_5T_6 + T_3T_5T_6 + T_4T_5T_6 + T_5^2T_6 + T_2T_6^2 +
    T_2T_5T_7 +  T_5^2T_7 + T_1T_6T_7 + T_6^2T_7 + T_3T_7^2 + T_5T_7^2 + T_6T_7^2 + T_7^3 + T_3T_5T_8 +
    T_2T_6T_8 + T_5T_7T_8 + T_7^2T_8 + T_4T_6 + T_5T_6 + T_6^2 + T_6T_7 + T_3 + T_5 + T_6 + T_7 + T_8$ \\
    
    $p_4(T) =$ & $T_1T_2T_4 + T_1T_2T_5 + T_2T_5^2 + T_3T_5^2 + T_1^2T_6 + T_1T_2T_6 + T_1T_5T_6 + T_2T_6^2 + T_1T_2T_7 +
    T_3T_4T_7 + T_4T_5T_7 + T_5^2T_7 + T_4T_6T_7 + T_4T_7^2 + T_3^2T_8 + T_3T_5T_8 + T_1T_6T_8 + T_2T_6T_8 +
    T_3T_6T_8 + T_5T_6T_8 + T_6^2T_8 + T_3T_7T_8 + T_4T_7T_8 + T_3T_8^2 + T_6T_8^2 + T_3T_5 + T_5^2 + T_5T_6 + 
    T_4T_7 + T_6T_7 + T_7^2 + T_5T_8$ \\
    
    $p_5(T) =$ & $T_2^3 + T_2T_3^2 + T_2T_3T_4 + T_1T_4^2 + T_2T_4^2 + T_4^3 + T_2^2T_5 + T_1T_2T_6 + T_2^2T_6 + T_4^2T_6 +
    T_1T_5T_6 + T_1T_6^2 + T_1T_2T_7 + T_4^2T_7 + T_1T_5T_7 + T_1T_6T_7 + T_4T_7^2 + T_5T_7^2 + T_6T_7^2 + T_7^3 
    T_2^2T_8 + T_2T_3T_8 + T_4^2T_8 + T_1T_5T_8 + T_3T_6T_8 + T_1T_8^2$ \\
    
    $p_6(T) =$ & $T_1T_3^2 + T_2^2T_4 + T_2T_4T_5 + T_1T_3T_6 + T_2T_4T_6 + T_3T_5T_6 + T_1T_3T_7 + T_4T_5T_7 + T_5^2T_7 +
    T_5T_6T_7 + T_5T_7^2 + T_1T_3T_8 + T_2T_4T_8 + T_3T_8^2 + T_5T_8^2 + T_6T_8^2 + T_7T_8^2 + T_8^3 + T_1T_3 +
    T_2T_6 + T_1T_7 + T_2T_8 + T_5T_8 + T_6T_8 + T_8^2$ \\
    
    $p_7(T) =$ & $T_1^2T_2 + T_3^3 + T_2^2T_5 + T_2T_3T_5 + T_3^2T_5 + T_4^2T_5 + T_2T_5^2 + T_1T_2T_6 + T_1T_3T_6 + T_3^2T_6 +
    T_2T_5T_6 + T_2T_6^2 + T_3T_6^2 + T_5T_6^2 + T_6^3 + T_3T_6T_7 + T_3T_7^2 + T_3^2T_8 + T_2T_5T_8 + T_3T_5T_8 +
    T_4T_5T_8 + T_6^2T_8 + T_3T_7T_8 + T_4T_7T_8 + T_5T_7T_8 + T_6T_7T_8 + T_7^2T_8 + T_1T_4 + T_4T_6 +
    T_1T_8 + T_6T_8$ \\
    
    $p_8(T) =$ & $T_1T_2T_3 + T_2T_4^2 + T_2T_3T_5 + T_4^2T_5 + T_3T_5^2 + T_2T_3T_6 + T_4^2T_6 + T_3T_5T_6 + T_2T_5T_7 +
    T_5^2T_7 + T_5T_6T_7 + T_4T_7^2 + T_5T_7^2 + T_6T_7^2 + T_7^3 + T_4^2T_8 + T_3T_5T_8 + T_2T_6T_8 + T_3T_6T_8 +
    T_1T_7T_8 + T_5T_7T_8 + T_3T_5 + T_5^2 + T_5T_6 + T_3T_7 + T_6T_7 + T_7^2 + T_5T_8 + T_7T_8$ \\
    
    $p_9(T) =$ & $T_1T_4^2 + T_3T_4^2 + T_1T_4T_5 + T_2T_4T_5 + T_4^2T_5 + T_4T_5^2 + T_1T_4T_6 + T_4^2T_6 + T_4T_5T_6 +
    T_1^2T_7 + T_2T_3T_7 + T_1T_4T_7 + T_3T_4T_7 + T_4^2T_7 + T_2T_5T_7 + T_3T_5T_7 + T_4T_5T_7 + T_5^2T_7 +
    T_1T_6T_7 + T_2T_6T_7 + T_3T_6T_7 + T_5T_6T_7 + T_2T_7^2 + T_3T_7^2 + T_5T_7^2 + T_1T_2T_8 + T_4^2T_8 +
    T_4T_5T_8 + T_5T_6T_8 + T_6^2T_8 + T_2T_7T_8 + T_1T_8^2$
\end{tabularx}
}

Let $x = (0, 1, 0, 1, 0, 0, 1, 1) \in A$ be a uniformly random element, which is obtained by measuring the given money state. Evaluating the Jacobian $J_{p_1, \dots, p_9}$ at $x$ gives

\[
J_{p_1, \dots, p_9}(x) = 
\begin{bmatrix}
    0 & 1 & 0 & 1 & 0 & 0 & 1 & 1 \\
    0 & 0 & 0 & 0 & 0 & 0 & 0 & 0 \\
    0 & 0 & 0 & 0 & 0 & 0 & 0 & 0 \\
    0 & 0 & 1 & 1 & 0 & 0 & 0 & 1 \\
    1 & 0 & 0 & 0 & 0 & 1 & 0 & 0 \\
    1 & 0 & 1 & 0 & 1 & 0 & 1 & 1 \\
    0 & 0 & 0 & 1 & 1 & 1 & 1 & 0 \\
    1 & 1 & 1 & 1 & 1 & 0 & 0 & 0 \\
    0 & 0 & 0 & 0 & 0 & 0 & 0 & 0 
\end{bmatrix}
\]

Now, the kernel of the above matrix gives us the subspace $A$.


\section{Quantum Money From Multivariate Hash Functions}

In this section, we investigate the security of the quantum money scheme proposed in \cite{zhandry2017qlightning}, which is based on multivariate hash functions. The hash functions used in this scheme are defined as follows. Let $m$ and $n$ be positive integers such that $m > n$. Define $\A = \{ \bm{A}_i \}_{1 \le i \le n}$ where each $\bm{A}_i \in \F_2^{m \times m}$ is an upper triangular matrix. The hash function corresponding to $\A$ is defined by
\[
\begin{array}{rrll}
     f_\A : &  \F_2^m & \longrightarrow & \F_2^n \\
     & x & \longmapsto & (x^T \bm{A}_1 x, \dots, x^T \bm{A}_n x).
\end{array}
\]
The function $f_\A$ is not collision resistant \cite{ding2007multivariates}, and in fact it is not hard to find collisions when the matrices $\bm{A}_i$ are random upper triangular.

The hardness assumption underlying this scheme is instead based on the multi-collision resistance of $f_\A$. Let us briefly define what that means. A set of $k + 1$ points of $\F_2^m$ is called non-affine if they form a $k$-dimensional affine space. A function $f$ is said to be $(k + 1)$-non-affine multi-collision resistant ($(k + 1)$-NAMCR) if it is hard to find $k + 1$ non-affine colliding inputs for $f$. The assumption is that for $k = \poly(n)$, $m < (k + 1 / 2)n$ and random upper triangular $\bm{A}_i$, the function $f_\A$ is $2(k + 1)$-NAMCR.

\subsection{The quantum money scheme}

In the following, we briefly explain the quantum money scheme of \cite{zhandry2017qlightning}. The public parameters of the scheme are:
\begin{itemize}
    \item The integers $m, n$. We set $n = \kappa$ where $\kappa$ is the security parameter. For the verification to work we can take $k = 2n$ and $m = kn = 2n^2$.
    \item The hash function $f_\A$. The upper triangular matrices in $\A = \{ \bm{A}_i \}_{1 \le i \le n}$ are generated uniformly at random.
\end{itemize}

A banknote in this scheme consists of a pair $(y, \ket{\bolt{y}})$ where $\ket{\bolt{y}}$, called a \textit{bolt}, is the money state and $y \in \F_2^n$ is the serial number. The bolt $\ket{\bolt{y}}$ is a product of states where each state is a superposition of all $x \in \F_2^m$ such that $f_\A(x) = y$. More precisely,
\[ \ket{\bolt{y}} \propto \ket{\bolt{y}'}^{\otimes (k + 1)}, \quad \text{where } \ket{\bolt{y}'} \propto \sum_{x : f_\A(x) = y} \ket{x}. \]
We refer the reader to \cite{zhandry2017qlightning} for the details of how a bolt is generated. Since we will use the verification circuit in our attack, we include some details here. To verify a bolt, one needs to only verify each of the states $\ket{\bolt{y}'}$ independently. The verification of an alleged banknote $\ket{\psi}$ proceeds in two steps: first, the state $\ket{\psi}$ is projected onto the span of the states $\{ \ket{\bolt{z}'} \}_{z \in \F_2^n}$. Then the function $f_\A$ is computed into an auxiliary register and measured. The result of the measurement is then compared against $y$ to determine the validity of the banknote.

The states $\{ \ket{\bolt{z}'} \}_{z \in \F_2^n}$ are orthogonal and span a linear subspace $B \subset \C^{2^m}$ of dimension $2^n$. To project onto $B$, a different set of basis states of $B$ are used.
\begin{fact}[\cite{zhandry2017qlightning}]
    Define the set of states
    \[ \ket{\phi_r} = \frac{1}{2^{m / 2}} \sum_{x \in \F_2^m} (-1)^{r \cdot f_\A(x)} \ket{x}, \quad r \in \F_2^n. \]
    Then the states $\{ \ket{\bolt{z}'} \}$ and $\{ \ket{\phi_r} \}$ span the same subspace $B$.
\end{fact}
Therefore, to project onto $B$ we need to project onto the span of the states $\ket{\phi_r}$. It is easy to prepare $\ket{\phi_r}$ given $r$. Conversely, it was shown in \cite{zhandry2017qlightning} that one can recover $r$ given the state $\ket{\phi_r}$. We record this result for sake of later reference.
\begin{theorem}[\cite{zhandry2017qlightning}]
    \label{thm:cmpt-r}
    There is a polynomial time quantum algorithm that, given the state $\ket{\phi_r}$, computes $r$ with overwhelming probability.
\end{theorem}

The verification procedure can now be summarized as follows. Given the input state $\ket{\psi}$, write $\ket{\psi} = \sum_r \alpha_r \ket{\phi_r} + \ket{\psi_1}$ where $\ket{\psi_1}$ is orthogonal to $B$. Then compute $r$ into an auxiliary register to obtain the state $\sum_r \alpha_r \ket{r}\ket{\phi_r} + \ket{\psi_2}$ for some state $\ket{\phi_2}$. Now uncompute the second register to obtain the state $\sum_r \alpha_r \ket{r} \ket{0} + \ket{\psi_3}$ for some state $\ket{\phi_3}$. The verification algorithm then measures the second register. If the measurement outcome is not zero the algorithm outputs `reject'. Otherwise, the post-measurement state is $\sum_r \alpha_r \ket{r} \ket{0}$, and the algorithm computes $\ket{\phi_r}$ into the second register to obtain the state $\sum_r \alpha_r \ket{r}\ket{\phi_r}$, and uncomputes $r$ to obtain the state $\sum_r \alpha_r \ket{\phi_r}$. Note that this state is now the projection of the original state $\ket{\psi}$ onto $B$. Finally, the algorithm computes $f_\A$ into another register and measures that register. If the measurement outcome is not equal to $y$, the algorithm outputs `reject', otherwise it outputs `accept'.

\subsection{Cryptanalysis}

In this section, we propose a polynomial time quantum algorithm that given a nonzero serial number $y \in \F_2^n$, generates a copy of the state
\[ \ket{\bolt{y}'} = \frac{1}{\sqrt{C_y}} \sum_{x : f_\A(x) = y} \ket{x}, \]
where $C_y = \#\{ x \in \F_2^m : f_\A(x) = y \}$, with probability $1$. 
We start by noting that \cite{zhandry2017qlightning}
\begin{equation}
    \label{equ:0-y-bolt}
    C_0\ket{\bolt{0}'} + C_y\ket{\bolt{y}'} \propto \sum_{r : r \cdot y = 0} \ket{\phi_r}.
\end{equation}
Let $\ket{\psi_y}$ be the normalized quantum state proportional to the state \eqref{equ:0-y-bolt}. Then $\ket{\psi_y}$ is a uniform superposition over the elements of $X_y = \{ x \in \F_2^m : f_\A(x) = 0 \text{ or } y \}$, i.e.,
\[ \ket{\psi_y} = \frac{1}{\sqrt{\#X_y}} \sum_{x \in X_y} \ket{x}. \]
Suppose there is a quantum algorithm $Q$ that, given $y$, can efficiently prepare $\ket{\psi_y}$. To generate a copy of $\ket{\bolt{y}'}$, perform the following steps. Call $Q$ to get a copy of $\ket{\psi_y}$. Compute $f_\A$ into an extra register to obtain the state
\begin{equation}
    \label{equ:0-y-bolt-1}
    \frac{1}{\sqrt{\#X_y}} \sum_{x \in X_y} \ket{x} \ket{f_\A(x)} = \frac{1}{\sqrt{\#X_y}} (\sqrt{C_0}\ket{\bolt{0}'}\ket{0} + \sqrt{C_y}\ket{\bolt{y}'}\ket{y}),
\end{equation}
and measure the second register. If the measurement outcome is equal to $y$ then we are done, otherwise repeat the process from the beginning. 

The probability of obtaining the state $\ket{\bolt{y}'}$ is $C_y / \#X_y$. Heuristically, the ratio $C_y / \#X_y$ is as large as a constant. In other words, there is no reason that $C_0$ is much larger or smaller than $C_y$ for any other $y \in \F_2^n$. We argue, with a little help from algebraic geometry, that this is indeed the case. The function $f_\A$ defines a morphism between two affine spaces $f_\A : \mathbb{A}_{\F_2}^m \rightarrow \mathbb{A}_{\F_2}^n$ where $\mathbb{A}_{\F_2}^m = \spec \F_2[T_1, \dots, T_m]$, and similarly for $\mathbb{A}_{\F_2}^n$. Since both $\mathbb{A}_{\F_2}^m$ and $\mathbb{A}_{\F_2}^m$ are geometrically irreducible, and since $0$ is the generic point of these spaces, it follows from the fibre dimension theorem \cite[Chapter 14 ]{gortz2010algebraic} that $\dim f_\A^{-1}(0) = \dim \mathbb{A}_{\F_2}^m - \dim \mathbb{A}_{\F_2}^n = m - n$\footnote{Here, the dimension is the geometric dimension not the vector space dimension}. We also have that for every $y \in \F_2^n$, all irreducible components of $f_\A^{-1}(y)$ have dimension at least $\dim f_\A^{-1}(0)$. Now, by the Lang-Weil theorem \cite{lang1954number}, \cite[Section 7.7]{poonen2017rational}, for any $y \in \F_2^n$ it holds that $\#f_\A^{-1}(y)(\F_2) = O(2^{\dim f_\A^{-1}(y)})$ where the constant depends on $f_\A$ but not on $y$. This means the ratio
\[ \frac{C_y}{\#X_y} = \frac{\#f_\A^{-1}(y)}{\#f_\A^{-1}(y) + \#f_\A^{-1}(0)} \]
is bounded below by a constant with overwhelming probability over the randomness of $\A$.

It only remains to show that algorithm $Q$ exists. The input to the algorithm is an element $y \in \F_2^n$ and the output is the state \eqref{equ:0-y-bolt}. The algorithm proceeds as follows. First, compute a basis for the space $\{ r \in \F_2^n : r \cdot y = 0 \}$. This can be efficiently done classically. Then using this basis, generate the superposition
\begin{equation}
    \label{equ:sup-y-ker}
    \frac{1}{\sqrt{2^{n - 1}}} \sum_{r : r \cdot y = 0} \ket{r}.
\end{equation}
Next, compute $\ket{\phi_r}$ into another register to obtain the state
\begin{equation}
    \label{equ:y-ker-phi}
    \frac{1}{\sqrt{2^{n - 1}}} \sum_{r : r \cdot y = 0} \ket{r}\ket{\phi_r}.
\end{equation}
Using Theorem \ref{thm:cmpt-r}, uncompute $r$ and discard the first register. This produces the state
\begin{equation}
    \label{equ:y-ker-uncmp}
    \ket{\psi_y} = \frac{1}{\sqrt{C}} \sum_{r : r \cdot y = 0} \ket{\phi_r},
\end{equation}
for an appropriate normalization constant $C$. The above algorithm is summarized as follows.

\begin{algorithmic}[1]
    \Require A serial number $y \in \F_2^n$
    \Ensure A bolt $\ket{\bolt{y}'}$
    \State Compute a set $G$ of generators for the kernel linear map $r \mapsto r \cdot y$.
    \State\label{step:bolt-gen-ker}Using $G$ prepare the superposition \eqref{equ:sup-y-ker}.
    \State Compute $\ket{\phi_r}$ into another register to obtain the state \eqref{equ:y-ker-phi}.
    \State Uncompute $r$ and discard the first register to obtain the state \eqref{equ:y-ker-uncmp}.
    \State Compute $f_\A$ into an extra register to obtain the state \eqref{equ:0-y-bolt-1}.
    \State Measure the second register. If the measurement outcome is equal to $y$ then return the post-measurement state. Otherwise, go to Step \ref{step:bolt-gen-ker}.
\end{algorithmic}


\section{Quantum Money from Quaternion Algebras}

In this section, we investigate the security of the quantum money scheme proposed in \cite{kane2021quantum}, which is based on quaternion algebras. The hard problem underlying this scheme is copying a quantum state that is encoding an eigenform of the Hecke operators. We give a reduction from this problem to a linear algebra problem, namely inverting a matrix with entries the representation numbers of quaternary quadratic forms. We believe that the latter problem is more suitable, and much more accessible, for future cryptanalysis.

Informally, the idea of our reduction is as follows. Suppose the quantum state to be copied is given by $\ket{\phi} = \sum \alpha_I \ket{I}$ where the sum is over the set $\{ I \}_{I \in \cl(\O)}$ of representatives of the left ideal classes of a maximal order $\O$ in a quaternion algebra. Let $M$ be the formal abelian group with basis $\cl(\O)$. For a fixed $J \in \cl(\O)$, based on a morphism from $M$ to the space of modular forms we obtain the identity
\[ f(q) = \sum_{I \in \cl(\O)} \frac{\alpha_I}{\alpha_J} \Theta_{I, J}(q), \]
where the left hand side is the modular form corresponding to the formal sum $\sum_{I \in \cl(\O)} \alpha_I [I]$, and the $\Theta_{I, J}$ are theta series corresponding to the ideals $I, J$. By equating the coefficients of different powers of $q$ in both sides of the above identity, we form a system of linear equations. The entries of the coefficient matrix $A$ associated with this system are the representation numbers from the theta series $\Theta_{I, J}$. The dimension of the system is equal to the size of the set $\cl(\O)$ which is exponentially large. We show that if one is able to approximate the operation $A^{-1}$ then one can make arbitrary many copies of $\ket{\phi}$.

\subsection{Background}

In this section, we review the minimal necessary background in modular forms and quaternion algebras, and set up some notations, for the following sections. Our main references for modular forms are \cite{diamond2005first, lang2012introduction, stein2007modular}. For a comprehensive treatment of quaternion algebras we refer the reader to \cite{voight2021quaternion}.

\subsubsection{Modular forms.}

Let $\H = \{ z \in \C : \mathrm{Im}(z) > 0 \}$ be the complex upper half plane, and let $\mathrm{SL}_2(\Z)$ be the special linear group over the integers. For $\gamma = \begin{bsmallmatrix} a & b \\ c & d \end{bsmallmatrix} \in \mathrm{SL}_2(\Z)$ and $z \in \H$ define $\gamma z = (az + b) / (cz + d)$, which defines a group action on $\H$. For any integer $N > 0$ denote by $\Gamma(N)$ the kernel of the projection $\mathrm{SL}_2(\Z) \rightarrow \mathrm{SL}_2(\Z / N\Z)$. A congruence subgroup of $\mathrm{SL}_2(\Z)$ of level $N$ is any subgroup $\Gamma$ containing $\Gamma(N)$. In this paper, we are particularly interested in the congruence subgroup
\[ \Gamma_0(N) = \left\{ \begin{bmatrix} a & b \\ c & d \end{bmatrix} \in \mathrm{SL}_2(\Z) : \begin{bmatrix} a & b \\ c & d \end{bmatrix} \equiv \begin{bmatrix} * & * \\ 0 & * \end{bmatrix} \pmod{N} \right\}. \] 
The weight-$k$ operator $[\gamma]_k$ on the functions $f \in \hom(\H, \C)$ is defined by $(f[\gamma]_k)(\tau) = (cz + d)^{-k} f(\gamma z)$. This is also a group action on $\hom(\H, \C)$ since $f [\gamma_1 \gamma_2]_k = (f [\gamma_1]_k) [\gamma_2]_k$. 

A weakly modular function of weight $k$ with respect to a congruence subgroup $\Gamma$ is a meromorphic function $f : \H \rightarrow \C$ such that $f[\gamma]_k = f$ for all $\gamma \in \Gamma$. Define the extended upper half plane $\H^*$ by adding the set of rational points of the projective line to $\H$, i.e., $\H^* = \H \cup \P^1(\Q)$. A \textit{modular form} of weight $k$ with respect to $\Gamma$ is a weakly modular function $f : \H \rightarrow \C$ that is holomorphic on $\H^*$. Every modular form $f$ has a Fourier expansion
\[ f(z) = \sum_{n = 0}^\infty a_n q^n, \quad q(z) = e^{2\pi iz} \]
called the $q$-expansion of $f$. If $a_0 = 0$ in the $q$-expansion of $f$, we say that $f$ is a \textit{cuspform}.

\subsubsection{Quaternion algebras.}

Let $F$ be a field of characteristic $\ne 2$. An $F$-algebra $B$ is called a quaternion algebra if $B$ has a basis $1, i, j, k$ as an $F$-vector space such that $i^2 = a, j^2 = b$ and $k = ij = -ji$ for some $a, b \in F^\times$. There is an involution operation $^-: B \rightarrow B$ defined by $t + xi + yj + zk \mapsto t - xi - yj - zk$. The reduced trace $\mathrm{trd}: B \rightarrow F$ and the reduced norm $\mathrm{nrd}: B \rightarrow F$ are defined by $\trd(\alpha) = \alpha + \overline{\alpha}$ and $\nrd(\alpha) = \alpha\overline{\alpha}$. For $F = \Q$, a quaternion algebra $B$ is said to be ramified at a prime $p$ if the completion $B \otimes_\Q \Q_p$ is a division ring, otherwise $B$ is unramified at $p$. In this paper, we are interested in the quaternion algebra $B$ over $\Q$, denoted by $\quat$, that is ramified at a single primes $p$ and at $\infty$. 

A $\Z$-lattice $I \subseteq \quat$ of rank $4$ is called a fractional ideal. A $\Z$-order (or simply an order) $\O \subset \quat$ is a fractional ideal that is also a ring. A maximal order is an order that is not properly contained in another order. For a fractional ideal $I \subseteq \quat$ define 
\[ \O_L(I) = \{ \alpha \in \quat : \alpha I \subseteq I \}. \]
Then $\O_L(I)$ is also an order and is called the left order of $I$. The right order of $I$ is defined similarly by $\O_R(I) = \{ \alpha \in \quat : I \alpha \subseteq I \}$. An ideal $I$ is invertible if there is another ideal $I'$ such that $II' = \O_L(I) = \O_R(I')$ and $I'I = \O_L(I') = \O_R(I)$. We denote the inverse of $I$, if it exists, by $I^{-1}$. The reduced norm of an ideal $I$ is a $\Z$-submodule of $\Q$ generated by the set $\{ \nrd(\alpha) : \alpha \in I \}$. The inverse of $I$ can also be written as $I^{-1} = \overline{I} \nrd(I)^{-1}$ where $\overline{I} = \{ \overline{\alpha} : \alpha \in I \}$ is the involution of $I$. For an order $\O \subseteq \quat$, A left fractional $\O$-ideal is a fractional ideal $I$ such that $\O \subseteq \O_L(I)$. Right fractional $\O$-ideals are defined similarly. 

Two fractional ideals $I, J$ are said to be in the same right class if $I = \alpha J$ for some $\alpha \in \quat^\times$, in which case we write $I \sim_R J$. The relation $\sim_R$ defines an equivalence relation on the set of fractional ideals, and a class of an ideal $I$ is denoted by $[I]_R$. Left equivalence classes are defined similarly. For an order $\O \subset \quat$, the right class set of $\O$ is defined as the isomorphism classes of invertible right $\O$-ideals, i.e.,  
\[ \cl_R(\O) = \{ [I]_R : I \text{ an invertible } \O\text{-ideal} \}. \]
Since the mapping $I \mapsto \overline{I}$ defines a bijection between left and right classes we often simply write $\cl(\O)$ instead of $\cl_R(\O)$.

\subsubsection{Hecke operators.}

Denote by $M_k(\Gamma_0(N))$ the space of modular forms of weight $k$ with respect to the subgroup $\Gamma_0(N)$, and denote by $S_k(\Gamma_0(N))$ the space of cuspforms of weight $k$ with respect to $\Gamma_0(N)$. For an integer $n$ with $(n, N) = 1$, the $n$-th Hecke operator on $M_k(\Gamma_0(N))$ is the linear operator 
\[
\begin{array}{rrll}
    T_n : & M_k(\Gamma_0(N)) & \longrightarrow & M_k(\Gamma_0(N)) \\
    & f(z) & \longmapsto & \displaystyle \frac{1}{n^{k - 1}} \sum_{\substack{ad = n \\ 0 \le b < d}} \frac{1}{d^k} f\left( \frac{az + b}{d} \right)
\end{array}
\]
We will always assume that $(n, N) = 1$, since this is what we are interested in, and also the theory is a bit simpler in this case. From the definition, we see that $T_n$ preserves the subspace $S_k(\Gamma_0(N))$. The Hecke operators for $M_k(\Gamma_0(N))$ satisfy 
\begin{equation}
    \label{equ:hecke-rec}
    \begin{aligned}
        T_{mn} & = T_m T_n && (m, n) = 1, \\
        T_{p^{r + 1}} & = T_{p^r}T_p - p^{k - 1}T_{p^{r - 1}} && p \text{ prime}.
    \end{aligned}
\end{equation}
For a modular form $f$, let $a_n(f)$ be the $n$-th coefficient of the $q$-expantion of $f$. Then for $f \in M_k(\Gamma_0(N))$ and prime $p$ we have
\begin{equation}
    a_n(T_p f) = 
    \begin{cases}
        a_{np}(f) + p^{k - 1}a_{n / p}(f) & \text{if } p \mid n, \\
        a_{np}(f) & \text{if } p \nmid n.
    \end{cases}
\end{equation}
It follows that for coprime integers $m, n$ we have $a_m(T_n f) = a_{mn}(f)$, and in particular, $a_1(T_n f) = a_n(f)$. Define $\T := \Z[\{ T_n \}_{n \ge 1}]$, the algebra generated by all the Hecke operators. Then $\T$ is a commutative ring called the \textit{Hecke algebra}. Since $\T$ is commutative, there are elements $f \in S_k(\Gamma_0(N))$ that are simultaneous eigenvectors for all element in $\T$. In particular, for any such $f$ we have $T_n f = \lambda_n f$ for some $\lambda_n \in \C^\times$. We refer to such $f$ as an \textit{eigenfunction}. If $f$ is normalized, i.e., $a_1(f) = 1$, it is called an \textit{eigenform}. The coefficient of the $q$-expansion of an eigenform satisfy important identities induced by \eqref{equ:hecke-rec}. More precisely, let $f = \sum_n a_nq^n$ be an eigenform, and assume $T_n f = \lambda_n f$. Then we have $a_n = \lambda_n$, and using \eqref{equ:hecke-rec} we obtain
\begin{equation}
    \label{equ:eigenform-coeff}
    \begin{aligned}
        a_{mn} & = a_m a_n && (m, n) = 1, \\
        a_{p^{r + 1}} & = a_{p^r}a_p - p^{k - 1}a_{p^{r - 1}} && p \text{ prime}.
    \end{aligned}
\end{equation}
Let $d\mu(z) = dx dy / y^2$, $z = x + iy \in \H$ be the hyperbolic measure, and let $V_{\Gamma_0(N)}$ be volume of the modular curve $X(\Gamma_0(N)))$ with respect to $d\mu(z)$. The space $S_k(\Gamma_0(N))$ can be made into an inner product space using the Petersson inner product
\[ \lrang{f, g} = \frac{1}{V_{\Gamma_0(N)}} \int_{X(\Gamma_0(N))} f(z) \overline{g(z)} (\mathrm{Im}(z))^k d\mu(z). \]
The Hecke operators $T_n$ are Hermitian with respect to this inner product, i.e., $\lrang{T_n f, g} = \lrang{f, T_ng}$. This means there exists a basis of $S_k(\Gamma_0(N))$ consisting of Hecke eigenforms.

One can also define a set of Hecke operators on ideal classes of a quaternion order as follows. Let $\O \subset \quat$ be an order and let $\cl(\O)$ be the (right) class set of $\O$. Define the formal divisor group 
\begin{equation}
    \label{equ:div-grp}
    M = \bigoplus_{I \in \cl(\O)} \Z.[I] 
\end{equation}
For an integer $n > 0$, $n \ne p$, the $n$-th Hecke operator on $M$ is defined by
\[
\begin{array}{rrll}
    T_n : & M & \longrightarrow & M \\
    & [I] & \longmapsto & \displaystyle \sum_{\substack{J \subseteq I \\ \nrd(JI^{-1}) = n}} [J],
\end{array}
\]
where the sum is over all invertible right $\O$-ideals contained in $I$. In the basis of the ideal classes in $\cl(\O)$, the Hecke operator $T_n$ is represented by the \textit{Brandt matrix} $B(n)$ which can be explicitly computed as follows. Let $h = \#\cl(\O)$ be the class number of $\O$, and fix a set $I_1, I_2, \dots, I_h$ of representatives of the ideal classes in $\cl(\O)$. Let $\O_i = \O_L(I_i)$, and let $w_i$ be the number of units in $\O_i$. The $n$-th Brandt matrix is defined by
\begin{equation}
    \label{equ:brandt-matrix}
    B(n)_{ij} = \frac{1}{w_i} \# \{ \alpha \in I_jI_i^{-1} : \nrd(\alpha) = n \nrd(I_jI_i^{-1}) \}.
\end{equation}
The Hecke operators $T_n$ also satisfy the relations \eqref{equ:hecke-rec}. Moreover, we can define an inner product on $M$ with respect to which the $T_n$ are Hermitian: $\lrang{[I_i], [I_j]} = \delta_{i,j} w_i / 2$ where $\delta_{i, j} = 1$ if $i = j$, and $\delta_{i, j} = 0$ otherwise. We will also call this inner product the Petersson inner product. Similar to the case of modular forms, there is a set of elements of $M_\R = M \otimes_\Z \R$ that are simultaneous eigenvectors for all Hecke operators. We also call these element Hecke eigenforms. The element
\[ \phi_E = \sum_{j = 1}^h 2w_j^{-1} [I_j] \]
of $M_\R$ is called the Eisenstein eigenform. We have $T_n \phi_E = \sigma'(n)$, where $\sigma'(n) = \sum_{d \mid n, (p, d) = 1} d$. The sub-$\T$-module of $M_\R$ orthogonal to $\phi_E$ is called the space of cuspforms. For any cuspform $\phi = \sum_{j = 1}^h \alpha_j [I_j]$ we have $\sum_{j = 1}^h \alpha_j = 0$.

\subsection{The quantum money scheme}
\label{sec:mod-mny-scheme}

Let us briefly review the quantum money scheme presented in \cite{kane2021quantum}. The parameters of the scheme are fixed as follows:
\begin{itemize}
    \item A prime $p \in O(2^\kappa)$ where $\kappa$ is the security parameter.
    \item The quaternion algebra $\quat$ over $\Q$ ramified at $p$ and $\infty$.
    \item A maximal order $\O \subset \quat$ of discriminant $p$.
    \item A set of Hecke operators $T_{n_j}$ where $n_j = \poly(\kappa)$, $1 \le j \le s$ and $s = \poly(\kappa)$.
\end{itemize}
Let $I_1, I_2, \dots, I_d$ be distinct representatives of the ideal classes in $\cl(\O)$, where $d = \#\cl(\O)$. In the above specific setup we have $d = \lfloor p / 12 \rfloor$. Let $M$ be the formal divisor group on $\cl(\O)$ defined in \eqref{equ:div-grp}. For simplicity, we will let $M_\C$ denote $M \otimes_{\Z} \C$. Let $\{ \phi_i \}_{1 \le i \le d}$ be a set of eigenforms in $M_\C$ that are a basis of $M_\C$, and let $\{ \ket{\phi_i} \}_{1 \le i \le d}$ be the corresponding quantum states. This means that for $\phi_i = \sum_i \alpha_i [I_i]$, we have
\[ \ket{\phi_i} = \frac{1}{\sqrt{\alpha}} \sum_{j = 1}^d \alpha_j \ket{I_j} \]
where $\alpha^2 = \sum_j \abs{\alpha_j}^2$. 
A banknote in the quantum money scheme consists of
\begin{enumerate}
    \item A money state of the form $\ket{\phi_i}\ket{\phi_i}$ for some $1 \le i \le d$,
    \item A serial number, which is a tuple $(b_j)_{1 \le j \le s}$ of approximate eigenvalues of the $T_{n_j}$ corresponding to the eigenform $\phi_i$.
    \item A classical digital signature of the serial number $(b_j)_{1 \le j \le s}$.
\end{enumerate}
The verification of a banknote is done by
\begin{enumerate}
    \item Verifying the signature of the serial number
    \item Using phase estimation on the operators $e^{iT_{n_j}}$ and the eigenstate $\ket{\phi_i}$ to recover an approximate serial number and compare against $(b_j)_{1 \le j \le s}$.
\end{enumerate}

It takes more than a few pages to give a precise account on the implementation of the above scheme. Here, we briefly explain two main points. 
\begin{description}[leftmargin = 0pt, font = \normalfont\itshape]
\item [Hamiltonian simulation.] To prepare an eigenstate $\ket{\phi_i}$ we start from the uniform superposition $\ket{\phi} = \sum_{I \in \cl(\O)} \ket{I}$ (normalization omitted). Then we use phase estimation on the operators $U_j = e^{iT_{n_j}}$ to project onto an eigenstate $\ket{\phi_i}$. Since the matrices of the operators $T_{n_j}$ are sparse, standard Hamiltonian simulation techniques can be used to implement the unitaries $U_j$. However, for an integer $t > 0$, the cost of simulating $U_j^t = e^{iT_{n_j}t}$ depends linearly on $t$. More precisely, the simulation cost depends linearly on $st \opnorm{T_{n_j}}_{\max}$ \cite{berry2015hamiltonian}, where $s$ is the sparsity degree of $T_{n_j}$, i.e., the maximum number of nonzero entries in any row, and $\opnorm{T_{n_j}}_{\max}$ is the largest entry of $T_{n_j}$ in absolute value. This means, using phase estimation, we can approximate the eigenvalues of the $T_{n_j}$ only to accuracy $1 / \poly(\kappa)$. 
\item [Projecting onto an eigenstate.] To be able to project onto a single eigenstate, we cannot use only one of the $T_{n_j}$, since in that case we would need an exponentially accurate phase estimation. However, if we run a polynomially accurate phase estimation for all the $T_{n_j}$, we might be able to project onto a random eigenstate $\ket{\phi_i}$. For this to hold, one needs to make some heuristic assumptions on the distribution of the eigenvalues of the Heck operators. Such distributions are already studied in the literature \cite{gun2008summation, murty2009effective, chow2015distinguishing}, but a formal treatment of the running time of the above projection based on these distributions is nontrivial. 
\end{description}

\subsection{Cryptanalysis}

In this section, we analyze the assumption behind the scheme in Section \ref{sec:mod-mny-scheme}. We use the connection between quaternion orders and modular forms to translate the problem to the space of modular forms. The connection is provided by the theory of theta series.

\subsubsection{From quaternions to modular forms.}

Let $m = 2k > 0$ be an integer and let $Q: \Z^m \rightarrow \Z$ be an integral positive definite quadratic form. The \textit{theta series} associated to $Q$ is a function $\Theta_Q: \H \rightarrow \C$ defined by
\[ \Theta_Q(z) = \sum_{x \in \Z^m} q^{Q(x)} = \sum_{n = 0}^\infty r_Q(n) q^n, \quad q = e^{2 \pi i z}. \]
The integers $r_Q(n) = \# \{ x \in \Z^m : Q(x) = n \}$ are called the \textit{representation numbers} of $Q$. Now, let $\O \subset \quat$ be a maximal order as in Section \ref{sec:mod-mny-scheme}. For any two ideals $I_i, I_j \subseteq \O$ define the mapping
\begin{equation}
    \label{equ:quad-form}
    \begin{array}{rrll}
        Q_{i, j} : & I_jI_i^{-1} & \longrightarrow & \Z \\
        & \alpha & \longmapsto & \displaystyle \nrd(\alpha) \nrd(I_iI_j^{-1})
    \end{array}
\end{equation}
This is a positive definite quadratic form. Using these quadratic forms we can define the morphism
\[
\begin{array}{rrll}
    \Theta : & M_\C \times M_\C & \longrightarrow & M_2(\Gamma_0(p)) \\
    & ([I_i], [I_j]) & \longmapsto & \displaystyle \sum_{\alpha \in I_jI_i^{-1}} q^{Q_{i, j}(\alpha)} = 1 + 2\sum_{n = 1}^\infty \lrang{T_n [I_i], [I_j]} q^n,
\end{array}
\]
where $q = e^{2 \pi i z}$, $z \in \H$. The last equality follows from the definition of Brandt matrices \eqref{equ:brandt-matrix}. Eichler \cite{eichler1973basis} proved that the spaces $M_\C$ and $M_2(\Gamma_0(p))$ are isomorphic as Hecke modules, see also \cite{kohel1999computing, kohel2001hecke}. More precisely, 
\begin{theorem}[Eichler]
    \label{thm:eichler-basis}
    The map $T_n \mapsto T_n$ of Hecke operators over $M_\C$ and $M_2(\Gamma_0(p))$ defines an isomorphism of Hecke algebras. Moreover, the morphism $\Theta$ defines a nondegenrate Hecke bilinear map,
    \[ \Theta(T_n [I_i], [I_j]) = \Theta([I_i], T_n [I_j]) = T_n \Theta([I_i], [I_j]), \]
    such that the traces of $T_n$ on $M_\C$ and $T_n$ on $M_2(\Gamma_0(p))$ agree.
\end{theorem}
If we fix one of the arguments of $\Theta$ we get a homomorphism $M_\C \rightarrow M_2(\Gamma_0(p))$. The following lemma explain what the image of $\Theta$ is for some specific arguments.
\begin{lemma}
    \label{lem:eigen-coeff}
    Let $\phi = \sum_{j = 1}^d \alpha_j [I_j]$ be an eigenform. Then for any integer $1 \le k \le d$ the modular form $\Theta([I_k], \phi)$ is a Hecke eigenfunction. Moreover, $\Theta([I_k], \phi) \ne 0$ if and only if $\alpha_k \ne 0$.
\end{lemma}
\begin{proof}
    The first part follows from Theorem \ref{thm:eichler-basis}. More precisely, for any $k$ and $n$
    \[ T_n \Theta([I_k], \phi) = \Theta([I_k], T_n \phi) = \Theta([I_k], \lambda_n \phi) = \lambda_n \Theta([I_k], \phi). \]
    For the second part, expanding $\Theta([I_k], \phi)$ gives
    \begin{align}
        \Theta([I_k], \phi) 
        & = \Theta\bigg([I_k], \sum_{j = 1}^d \alpha_j [I_j] \bigg) \nonumber \\
        & = \sum_{j = 1}^d \alpha_j \Theta([I_k], [I_j]) && (\text{by linearity of } \Theta) \nonumber \\
        & = \sum_{j = 1}^d \alpha_j \Big( 1 + 2\sum_{n = 1}^\infty \lrang{T_n [I_k], [I_j]} q^n \Big) && (\text{by definition}) \nonumber \\
        & = \sum_{j = 1}^d \alpha_j + 2 \sum_{n = 1}^\infty \lrang{T_n [I_k], \phi} q^n && (\text{by linearity of } \lrang{\,,}) \label{equ:theta-expand}
    \end{align}
    If $\phi$ is the Eisenstein eigenform $\phi_E$ then $\alpha_k \ne 0$ and $\lrang{T_n [I_k], \phi} = \sigma'(n)\alpha_k \ne 0$, so $\Theta([I_k], \phi) \ne 0$ and the statement is true. So assume that $\phi$ is a cuspform. If $\alpha_k = 0$ then 
    \begin{equation}
        \label{equ:bilin-zero}
        \lrang{T_n [I_k], \phi} = \lrang{[I_k], T_n \phi} = \lambda_n \lrang{[I_k], \phi} = \lambda_n \alpha_k = 0
    \end{equation}
    for all $n \ge 1$. It follows from \eqref{equ:theta-expand} and \eqref{equ:bilin-zero} that $\Theta([I_k], \phi) = 0$. If $\alpha_k \ne 0$ then $\lrang{[I_k], \phi} \ne 0$ and again we see from \eqref{equ:theta-expand} that $\Theta([I_k], \phi) \ne 0$.
\end{proof}
A direct consequence of Lemma \ref{lem:eigen-coeff} is that the dimension of the image of the homomorphism $\Theta([I_k], -): M_\C \rightarrow M_2(\Gamma_0(p))$ depends on the coefficients $\lrang{[I_k], \phi_j}$ being zero or not for different eigenforms $\phi_j$. In fact, using the technique in \cite{zbMATH07132848}, one can prove that
\[ \dim_\C \Theta([I_k], M_\C) = \abs{\Sigma(k)} \]
where $\Sigma(k) = \{ j : \lrang{[I_k], \phi_j} \ne 0 \}$. Therefore, one would expect that for a random $1 \le k \le d$, $\lrang{[I_k], \phi} \ne 0$ with high probability, or equivalently, that $\dim_\C \Theta([I_k], M_\C)$ is not too small. We now outline our reduction.

\subsubsection{The reduction}

Suppose we are given the quantum state
\[ \ket{\phi} = \frac{1}{\sqrt{\alpha}} \sum_{j = 1}^d \alpha_j \ket{I_j} \] 
which represents a random eigenform $\phi = \sum_{j = 1}^d \alpha_j [I_j]$. Since $\phi_j$ is random, it is a cuspform with overwhelming probability. For a random $k$, we assume, by the above remarks, that $\alpha_k \ne 0$. Let $f(q) = \Theta([I_k], \phi) / \alpha_k \in M_2(\Gamma_0(p))$. Then, by Lemma \ref{lem:eigen-coeff}, $f(q)$ is a Hecke eigenfunction. A closer look at \eqref{equ:theta-expand} shows that $f(q)$ is in fact an eigenform with eigenvalues the same as the eigenvalues of $\phi$. Therefore, we have the identity
\begin{equation}
    \label{equ:eigenform-idty}
    f(q) = \sum_{j = 1}^d \frac{\alpha_j}{\alpha_k} \Theta([I_k], [I_j])
\end{equation}
of modular forms. Let $f(q) = \sum_{n = 1}^\infty a_n q^n$ so that $a_1, a_2, \dots$ are the eigenvalues of $f$. The idea is to build a system of linear equations, where $\alpha_j$ are the unknowns, by equating the coefficients of $q^n$ on both sides of \eqref{equ:eigenform-idty} for different values of $n$. For any $n \ge 1$ we get an equation
\begin{equation}
    \label{equ:eigen-equ-lin}
    a_n = \frac{\alpha_1}{\alpha_k} \lrang{T_n [I_k], [I_1]} + \frac{\alpha_2}{\alpha_k} \lrang{T_n [I_k], [I_2]} + \cdots + \frac{\alpha_d}{\alpha_k} \lrang{T_n [I_k], [I_d]}.
\end{equation}
Let $n_1, n_2, \dots, n_d$ be a set of positive integers. Gathering linear equations of the form \eqref{equ:eigen-equ-lin} for all the $n_i$ we obtain a system
\begin{equation}
    \label{equ:lin-sys}
    A \bm{\alpha} = \bm{a}
\end{equation}
where $\bm{a}$ is the column vector $[a_{n_1}, a_{n_2}, \dots, a_{n_d}]^T$, $\bm{\alpha}$ is the column vector $[\alpha_1 / \alpha_k, \alpha_2 / \alpha_2, \dots, \alpha_d / \alpha_k]^T$ and the matrix $A = (\lrang{T_{n_i} [I_k], [I_j]})_{i, j}$. The quantum state representing the vectors $\bm{a}$ and $\bm{\alpha}$ are 
\begin{align*}
    \ket{\bm{a}} & = \frac{1}{\sqrt{\sum_{j = 1}^d a_{n_j}^2}} \sum_{j = 1}^d a_{n_j} \ket{j} \text{ and} \\
    \ket{\bm{\alpha}} & =\frac{1}{\sqrt{\sum_{j = 1}^d \alpha_j^2 / \alpha_k^2}} \sum_{j = 1}^d \frac{\alpha_j}{\alpha_k} \ket{I_j} = \frac{1}{\sqrt{\alpha}} \sum_{j = 1}^d \alpha_j \ket{I_j} = \ket{\phi}, 
\end{align*}
respectively. Suppose there is a quantum algorithm $\mathcal{A}$ that can efficiently approximate the operation $A^{-1}$. Then one could compute an approximate copy of $\ket{\phi}$ as $\ket{\phi} = A^{-1} \ket{\bm{a}}$. Therefore, assuming access to $\mathcal{A}$, the problem is reduced to the following
\begin{problem}
    \label{prb:init-state}
    Generate (an approximation of) the state $\ket{\bm{a}}$ for an appropriate choice of distinct integers $n_1, n_2, \dots, n_d$.
\end{problem}
The main obstacle to solving Problem \ref{prb:init-state} is that since $d$ is exponentially large, the $n_i$ will be exponentially large as $i$ becomes large regardless of which set of the $n_i$ we choose. Recall that $a_{n}$ is an eigenvalue of $\ket{\phi}$ for any $n \ge 1$. The best classical algorithm for computing $a_n$ has complexity $\poly(p \log n)$ \cite{couveignes2011computational}, which is polynomial in $\log n$ but exponential in $\log p = O(\kappa)$, where $\kappa$ is the security parameter. Another possible approach is to use phase estimation to directly compute $a_n$, since we know that
\[ e^{iT_n} \ket{\phi} = e^{i a_n} \ket{\phi}. \]
This would also fail because of the following reason. It is known that $\abs{a_n} \le \sigma_0(n) \sqrt{n}$ \cite{deligne1974conjecture}, where $\sigma_0(n)$ is the number of positive divisors of $n$. However, the value $\abs{a_n}$ can be exponentially large in $\kappa$. This means we need to perform phase estimation with exponential accuracy. But we can only simulate $e^{iT_nt}$ for time $t$ that is bounded by $\poly(\kappa)$. 

To solve this problem, the idea is to use phase estimation but for a specific set of integers $n_i$ such that we can exploit the relations \eqref{equ:eigenform-coeff} between the $a_{n_i}$'s. More precisely, let $s = \lceil \log d \rceil$ and let $C = \{ \ell_1, \ell_2, \dots, \ell_s \}$ be a set $s$ distinct primes of size $\poly(\kappa)$, say the first $s$ primes $2, 3, 5, \dots, \ell_s$. Then we choose the $n_i$ to be the set
\begin{equation}
    \label{equ:n_i}
    \{ n_i \}_{1 \le i \le d} := \{ \ell_1^{b_1} \ell_2^{b_2} \cdots \ell_s^{b_s} : (b_1, \dots, b_s) \in \{ 0, 1 \}^s \}.
\end{equation}
Explicitly, if $j = b_1b_2 \cdots b_s$ is the binary representation of $j$ then we set $n_j = \ell_1^{b_1} \ell_2^{b_2} \cdots \ell_s^{b_s}$.

\begin{theorem}
    Given the quantum state $\ket{\phi}$ representing an eigenform $\phi$, let $f(q) = \sum_{n = 1}^\infty a_nq^n$ be the modular form, defined by \eqref{equ:eigenform-idty}, that corresponds to $\phi$. Let $\{ n_i \}_{1 \le i \le d}$ be the set of integers defined in \eqref{equ:n_i}, and let
    \[ \ket{\bm{a}} = \frac{1}{\sqrt{\sum_{j = 1}^d a_{n_j}^2}} \sum_{j = 1}^d a_{n_j} \ket{j}. \]
    For any constant $c > 0$, there is a polynomial time quantum algorithm that can prepare a state $\ket{\tilde{\bm{a}}}$ such that $\opnorm{\ket{\bm{a}}\bra{\bm{a}} - \ket{\tilde{\bm{a}}}\bra{\tilde{\bm{a}}}}_1 \le 1 / \kappa^c$.
\end{theorem}
\begin{proof}
    For any prime $\ell \in C$ we can use phase estimation on the operator $e^{iT_\ell}$ and the eigenstate $\ket{\phi}$ to compute an approximate eigenvalue $\tilde{a}_\ell$ such that $\abs{a_\ell - \tilde{a}_\ell} \le 1 / s \kappa^c$. Now, using the single qubit operation
    \[ \tilde{R}_\ell = \frac{1}{\sqrt{1 + \smash[b]{\tilde{a}_\ell^2}}}
    \begin{bmatrix}
        1 & -\tilde{a}_\ell \\
        \tilde{a}_\ell & 1
    \end{bmatrix}
    \]
    we can generate the single qubit state
    \[ \ket{\tilde{\psi}_\ell} := \tilde{R}_\ell \ket{0} = \frac{1}{\sqrt{1 + \smash[b]{\tilde{a}_\ell^2}}} (\ket{0} + \tilde{a}_\ell \ket{1}). \]
    Tensoring these states together for all $\ell \in C$ we obtain the state $\ket{\tilde{\bm{a}}} := \bigotimes_{\ell \in C} \ket{\tilde{\psi}_\ell}$. Now, for any $n_j$ we can write $n_j = \ell_1^{b_1} \ell_2^{b_2} \cdots \ell_s^{b_s}$ where $b_1b_2 \cdots b_s$ is the binary representation of $j$. Therefore, using the relations \eqref{equ:eigenform-coeff} we have $a_{n_j} = a_{\ell_1}^{b_1} \cdots a_{\ell_s}^{b_s}$. So we can rewrite the state $\ket{\bm{a}}$ as 
    \begin{align*}
        \ket{\bm{a}}
        & = \bigg( \prod_{j = 1}^s (1 + a_\ell^2) \bigg)^{-1 / 2} \sum_{(b_1, \dots, b_s) \in \{ 0, 1 \}^s} a_{\ell_1}^{b_1} \cdots a_{\ell_s}^{b_s} \ket{b_1, \dots, b_s} \\
        & = \bigotimes_{\ell \in C} \frac{1}{\sqrt{1 + \smash[b]{a_\ell^2}}} (\ket{0} + a_\ell \ket{1})
    \end{align*}
    Let $\ket{\psi_\ell} = (\ket{0} + a_\ell \ket{1}) / \sqrt{1 + \smash[b]{a_\ell^2}}$. Then
    \[ \opnorm{\ket{\psi_\ell}\bra{\psi_\ell} - \ket{\tilde{\psi}_\ell}\bra{\tilde{\psi}_\ell}}_1 = \sqrt{1 - \abs{\braket{\psi_\ell}{\tilde{\psi}_\ell}}^2} \le \frac{1}{s\kappa^c}, \]
    and we have
    \begin{align*}
        \opnorm{\ket{\bm{a}}\bra{\bm{a}} - \ket{\tilde{\bm{a}}}\bra{\tilde{\bm{a}}}}_1
        & = \opnorm[\bigg]{\bigotimes_{\ell \in C}\ket{\psi_\ell}\bra{\psi_\ell} - \bigotimes_{\ell \in C}\ket{\tilde{\psi}_\ell}\bra{\tilde{\psi}_\ell}}_1 \\
        & \le \sum_{\ell \in C} \opnorm{\ket{\psi_\ell}\bra{\psi_\ell} - \ket{\tilde{\psi}_\ell}\bra{\tilde{\psi}_\ell}}_1 \\
        & \le \frac{1}{\kappa^c} \qedhere
    \end{align*}
\end{proof}

\newpage
\bibliographystyle{plain}
\bibliography{references}

\end{document}